\documentclass[a4paper]{article}
\usepackage[utf8]{inputenc}

\usepackage[margin=1.3in]{geometry}

\usepackage{amsthm,amssymb,amsfonts,amsmath}
\usepackage{amssymb,amsfonts,amsmath}
\usepackage{xcolor}
\newtheorem{theorem}{Theorem}
\newtheorem{lemma}[theorem]{Lemma}
\newtheorem{corollary}[theorem]{Corollary}

\newtheorem{definition}{Definition}

\newtheorem{remark}{Remark}

\usepackage{graphicx}
\graphicspath{{figures/}}

\usepackage{thm-restate}
\usepackage{url}
\usepackage{faktor}
\usepackage{tikz-cd}
\usepackage{authblk}

\title{The complexity of sharing a pizza}

\author[1]{Patrick Schnider\thanks{The author has received funding from the European Research Council under the European Unions Seventh Framework Programme ERC Grant agreement ERC StG 716424 - CASe.}}
\affil[1]{Department of Mathematical Sciences, University of Copenhagen, Denmark\\ \texttt{ps@math.ku.dk}}
        
\date{}

\begin{document}

\maketitle

\begin{abstract}
Assume you have a 2-dimensional pizza with $2n$ ingredients that you want to share with your friend.
For this you are allowed to cut the pizza using several straight cuts, and then give every second piece to your friend.
You want to do this fairly, that is, your friend and you should each get exactly half of each ingredient.
How many cuts do you need?

It was recently shown using topological methods that $n$ cuts always suffice.
In this work, we study the computational complexity of finding such $n$ cuts.
Our main result is that this problem is PPA-complete when the ingredients are represented as point sets.
For this, we give a new proof that for point sets $n$ cuts suffice, which does not use any topological methods.

We further prove several hardness results as well as a higher-dimensional variant for the case where the ingredients are well-separated.
\end{abstract}

\section{Introduction}

\subsection{Mass partitions}

The study of \emph{mass partitions} is a large and rapidly growing area of research in discrete and computational geometry.
It has its origins in the classic \emph{Ham-Sandwich theorem} \cite{HS}.
This theorem states that any $d$ \emph{mass distributions} in $\mathbb{R}^d$ can be simultaneously bisected by a hyperplane.
A mass distribution $\mu$ in $\mathbb{R}^d$ is a measure on $\mathbb{R}^d$ such that all open subsets of $\mathbb{R}^d$ are measurable, $0<\mu(\mathbb{R}^d)<\infty$ and $\mu(S)=0$ for every lower-dimensional subset $S$ of $\mathbb{R}^d$.
A vivid example of this result is, that it is possible to share a 3-dimensional sandwich, consisting of bread, ham and cheese, with a friend by cutting it with one straight cut such that both will get exactly half of each ingredient.
This works no matter how the ingredients lie.
In fact, as Edelsbrunner puts it, this even works if the cheese is still in the fridge \cite{EdelsbrunnerBook}.

But what if there are more ingredients, for example on a pizza?
One way to bisect more than $d$ masses is to use more complicated cuts, such as algebraic surfaces of fixed degree \cite{HS} or piece-wise linear cuts with a fixed number of turns \cite{Chess1,WedgesArxiv}.
Another option is to use several straight cuts, as introduced by Bereg et al.~\cite{Bereg}:
Consider some arrangement of $n$ hyperplanes in $\mathbb{R}^d$.
The cells of this arrangement allow a natural 2-coloring, where two cells get a different color whenever they share a $(d-1)$-dimensional face.
We say that an arrangement bisects a mass distribution $\mu$ if the cells of each color contain exactly half of $\mu$.
See Figure \ref{fig:partition_example} for an illustration.
It was conjectured by Langerman that any $nd$ mass distributions in $\mathbb{R}^d$ can be simultaneously bisected by an arrangement of $n$ hyperplanes (\cite{Stefan}, see also \cite{PizzaArxiv}).
In a series of papers, this conjecture has been resolved for $4$ masses in the plane \cite{PizzaArxiv}, for any number of masses in any dimension that is a power of 2 (and thus in particular also in the plane) \cite{Pizza1} and in a relaxed setting for any number of masses in any dimension \cite{SubspacesDCG}.
However, the general conjecture remains open.

\begin{figure}
\centering
\includegraphics[scale=0.7]{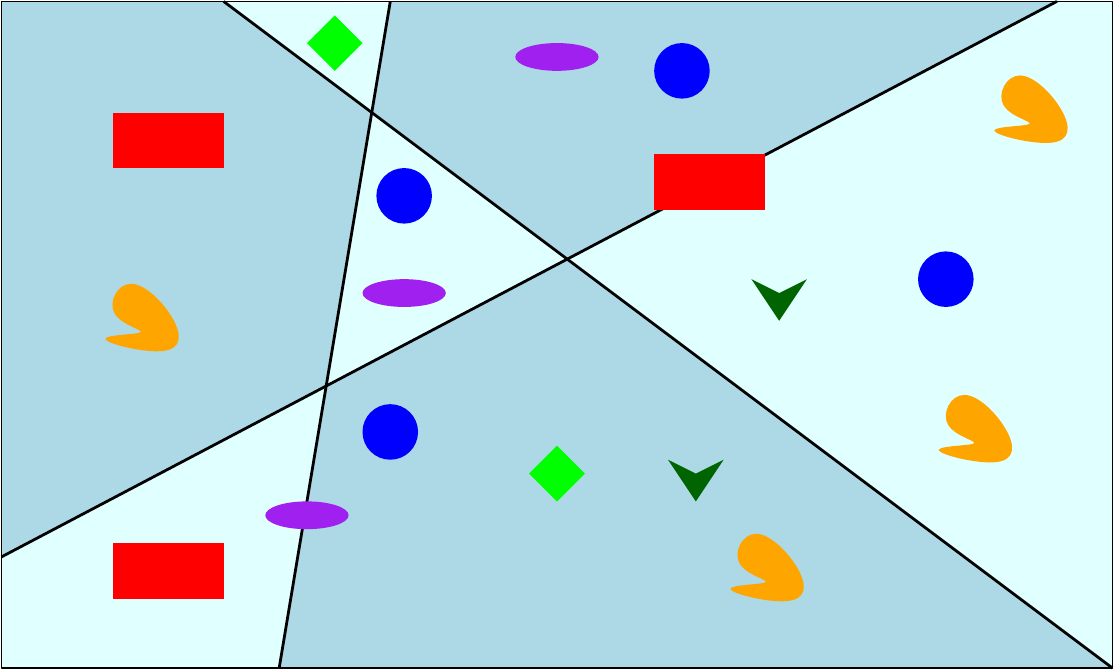}
\caption{A bisection of 6 masses with 3 cuts.}
\label{fig:partition_example}
\end{figure}

There are many other variants of mass partitions that have been studied, see e.g.~\cite{KanoSurvey, SoberonSurvey} for recent surveys.
In this work, we focus on the algorithmic aspects of the 2-dimensional variant of bisections with hyperplane arrangements, that is, bisections with line arrangements.
For this, let us formally define the involved objects.

\begin{definition}[Partition induced by an arrangement of oriented lines]
Let $A=(\ell_1,\ldots,\ell_n)$ be a set of oriented lines in the plane.
For each $\ell_i$, define be $R_i^+$ and $R_i^-$ the part of the plane on the positive and negative side of $\ell_i$, respectively.
Define $R^+(A)$ as the part of the plane lying in an even number of $R_i^+$ and not on any of the $\ell_i$.
Similarly, define $R^-(A)$ as the part of the plane lying in an odd number of $R_i^+$ and not on any of the $\ell_i$.
Now, $R^+(A)$ and $R^-(A)$ are disjoint, and they partition $\mathbb{R}^2\setminus\{\ell_1,\ldots,\ell_n\}$ into two parts.
\end{definition}

Note that reorienting one line just swaps $R^+(A)$ and $R^-(A)$, so up to symmetry, the two sides are already determined by the underlying unoriented line arrangement.
We will thus often forget about the orientations and just say that a mass is bisected by a line arrangement.
Hubard and Karasev \cite{Pizza1} have shown the following:

\begin{theorem}[Planar pizza cutting theorem \cite{Pizza1}]
Any $2n$ mass distributions in the plane can be simultaneously bisected by an arrangement of $n$ lines.
\end{theorem}

From an algorithmic point of view, we want to restrict our attention to efficiently computable mass distributions.

\begin{definition}[Computable mass distribution]
A computable mass distribution is a continuous function $\mu$ which assigns to each arrangement of $n$ oriented lines two values $\mu(R^+(A))$ and $\mu(R^-(A))$, such that $\mu(R^+(A))+\mu(R^-(A))=\mu(R^+(A'))+\mu(R^-(A'))$ for any two arrangements $A$ and $A'$.
We further assume that $\mu$ can be computed in time polynomial in the description of the input arrangement.
\end{definition}

We now have that the following problem always has a solution.

\begin{definition}[\textsc{PizzaCutting}]
The problem \textsc{PizzaCutting} takes as input $2n$ computable mass distributions $\mu_1,\ldots,\mu_{2n}$ and returns an arrangement $A$ of oriented lines in the plane such that for each $i$ we have $\mu_i(R^+(A))=\mu_i(R^-(A))$.
\end{definition}

An important special case of masses are point sets with the counting measure.
They do not quite fit the above framework of mass distributions, as the number of points on a line can be non-zero.
This can however be resolved by rounding: we say that a line arrangement bisects a point set $P$, if there are at most $\lfloor\frac{|P|}{2}\rfloor$ many points of $P$ in both parts.
Note that if the number of points in $P$ is odd, this implies that at least one point needs to lie on some line.
In the following, we will assume that all point sets are in general position, that is, no three points lie on a common line.
With this, we can assume that for each point set at most one point lies on a line.
Standard arguments (see e.g.~\cite{MatousekBU}) show that the existence of partitions for mass distributions imply the analogous result for point sets with this definition of bisection.
Alternatively, we give a direct proof of the following in Section \ref{sec:containment}.

\begin{corollary}[Discrete planar pizza cutting theorem]\label{cor:discreteplanar}
Any $2n$ point sets in the plane can be simultaneously bisected by an arrangement of $n$ lines.
\end{corollary}

Thus, also the following discrete version always has a solution.

\begin{definition}[\textsc{DiscretePizzaCutting}]
The problem \textsc{DiscretePizzaCutting} takes as input $2n$ point sets $P_1,\ldots,P_{2n}$ in general position in the plane and returns an arrangement $A$ of oriented lines in the plane such that for each $i$ we have $|P_i\cap R^+(A)|=|P_i\cap R^-(A)|=\lfloor\frac{|P_i|}{2}\rfloor$.
\end{definition}

The pizza cutting problem can be viewed as a higher-dimensional generalization of the consensus halving and necklace splitting problems.

\begin{definition}[\textsc{ConsensusHalving/NecklaceSplitting}]
The problem \textsc{ConsensusHalving} takes as input $n$ valuation functions $v_1,\ldots,v_{n}$ on the interval $[0,1]$ and returns a partition of $[0,1]$ into $n+1$ intervals (that is, using $n$ cuts), each labeled $"+"$ or $"-"$, such that for each valuation function we have $v_i(\mathcal{I}^+)=v_i(\mathcal{I}^-)$ (where $\mathcal{I}^x$ denotes the union of intervals labeled $"x"$).
The problem \textsc{NecklaceSplitting} is the same, but taking as input $n$ point sets, again using the above definition of bisections of point sets.
\end{definition}

Again, a solution to the problems is always guaranteed to exist.
In the case of mass distributions, this result is known as the \emph{Hobby-Rice theorem} \cite{Necklace_Hobby}.
For necklaces, the statement holds even for the generalized problem of sharing with more than two people \cite{Necklace_Alon,Necklace_Alon_West}.
In this work, whenever we refer to the \emph{Necklace splitting theorem}, we mean the version for two people.

Finally, for all above problems, we can also consider the decision version, where we are given one more measure or point set than the number that can always be bisected, and we need to decide whether there still is a bisection.
We denote these problems by adding "\textsc{Decision}" to their name.

\subsection{Algorithms and complexity}

Most proofs of existence of certain mass partitions use topological methods, which, by their nature, are not algorithmic.
Thus, there has been quite some effort in developing algorithms that find these promised partitions, ideally efficiently.
Arguably the most famous result in this direction are the algorithms for Ham-Sandwich cuts by Lo, Steiger and Matou\v{s}ek \cite{HSAlgo2D,HSAlgo}.
While in the plane, their algorithm runs in linear time, in general the runtime shows an exponential dependency on the dimension.
This curse of dimensionality seems to be a common issue for many algorithmic version of mass partition problems, and most problems have only been studied from an algorithmic point of view in low dimensions, where the constructed algorithms either rely on a relatively small space of solutions or a simplified proof which allows for an algorithmic formulation, see e.g.\ \cite{AichholzerIslands,BeregFans,LinesInSpace}.

The curse of dimensionality was made explicit for the first time by Knauer, Tiwary and Werner, who showed that deciding whether there is a Ham-Sandwich cut through a given point in arbitrary dimensions is W[1]-hard (and thus also NP-hard) \cite{Knauer}.
More recently, in several breakthrough papers, Filos-Ratsikas and Goldberg have shown that computing Ham-Sandwich cuts in arbitrary dimensions is PPA-complete, and so are \textsc{NecklaceSplitting} and \textsc{ConsensusHalving}, the latter even in an approximation version \cite{GoldbergConsensus,GoldbergNecklace}.

The class PPA was introduced in 1994 by Papadimitriou \cite{Papadimitriou}.
It captures search problems, where the existence of a solution is guaranteed by a parity argument in a graph.
More specifically, the defining problem is the following search problem in a (potentially exponentially sized) graph $G$: given a vertex of odd degree in $G$, where $G$ is represented via a polynomially-sized circuit which takes as input a vertex and outputs its neighbors, find another vertex of odd degree.
The class PPA is a subclass of TFNP (Total Function NP), which are total search problems where solutions can be verified efficiently.

\begin{figure}
\centering
\includegraphics[scale=0.7]{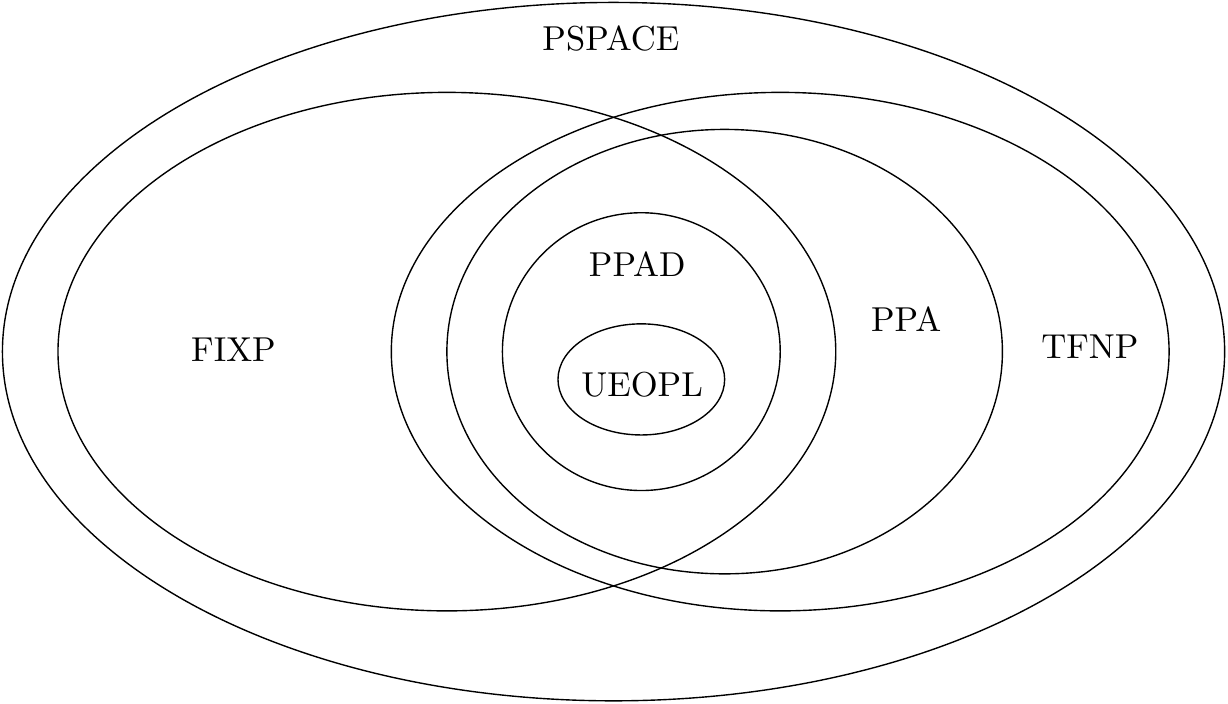}
\caption{Containment relations of some complexity classes for total problems.}
\label{fig:classes}
\end{figure}

A subclass of PPA that is of importance in this work is UEOPL \cite{UEOPL}, which is a subclass of PPAD \cite{Papadimitriou}.
PPAD is similar to PPA, but instead of an undirected graph, we are given a directed graph in which each vertex has at most one predecessor and at most one successor.
We are given a vertex without predecessor, and our goal is to find another vertex without predecessor or successor.
If we are further given a potential function which strictly increases on a directed path such that there is a unique vertex with maximal potential, finding this vertex is the defining problem for the class UEOPL.

The class UEOPL is related to mass partitions through the fact that finding the unique  discrete Ham-Sandwich cut in the case that the point sets are \emph{well-separated} is in UEOPL \cite{AlphaHS}
We say that $k$ point sets $P_1,\ldots,P_k$ in $\mathbb{R}^d$ are well-separated, if for no $d$-tuple of them their convex hulls can be intersected with a $(d-2)$-dimensional affine subspace.\footnote{Admittedly, this would be a very weird pizza.}
This definition extends to masses by forbidding intersections of affine subspaces with the convex hulls of their supports.
In fact, for well-separated masses and point sets, the \emph{$\alpha$-Ham-Sandwich theorem} states that it is always possible to simultaneously cut off an arbitrary given fraction from each mass or point set with a single hyperplane \cite{AlphaHSMasses,AlphaHSPoints}.

The class PPAD has so far mostly been related to the computation of Nash and market equilibria \cite{Equi1,Equi2,Equi3,Equi4,Equi5,Equi6,Equi7,Equi8,Equi9,Equi10,Equi11}.

Finally, the two last classes that are relevant for this work are FIXP and $\exists\mathbb{R}$.
The class FIXP is the class of problems that can be reduced to finding a Brouwer fixed point \cite{FIXP}, whereas $\exists\mathbb{R}$ is the class of decision problems which can be written in the existential theory of the reals.

In the context of mass partitions, apart from the above mentioned results on Ham-Sandwich cuts, Consensus halvings and Necklace splittings, some of the above classes also appear in the complexity of \emph{Square-Cut Pizza sharing}.
In this variant, introduced by Karasev, Rold\'{a}n-Pensado and Sober\'{o}n, $n$ masses in the plane are bisected by a cut which is a union of at most $n$ axis-parallel segments, or in other words, a piecewise linear cut with at most $n-1$ $90^{\circ}$-turns \cite{Chess1}.
It was recently shown by Deligkas, Fearnley and Melissourgos that finding such a cut even for restricted inputs is PPA-complete, finding a cut where a constant number of additional turns are allowed is PPAD-hard, and that the corresponding decision version is NP-hard \cite{SCPizza}.
For more general masses, they show the problem to be FIXP-hard and the decision version to be $\exists\mathbb{R}$-complete.
This work was heavily influenced and motivated by their paper: apart from the different setting, their results are very similar from the results in this work, showing the relation between those two pizza cutting variants.
Indeed, in a new version of their paper, the authors of \cite{SCPizza} prove some of the same hardness results as this work, and also some stronger hardness results for approximate bisections.

\subsection{Our contributions}

Our main contribution is that \textsc{DiscretePizzaCutting} is PPA-complete.
While the hardness is rather straight-forward, the containment requires some more work.
We give a new proof for the existence of pizza cuttings for point sets which, while inspired by the ideas of Hubard and Karasev, uses only elementary geometric techniques which allow us to place the problem in PPA.

We further prove that \textsc{PizzaCutting} is FIXP-hard and that \textsc{PizzaCuttingDecision} is $\exists\mathbb{R}$-hard and that finding the minimum number of cuts required to bisect an instance of \textsc{DiscretePizzaCutting} is NP-hard.

Finally, for well-separated masses, we show that the $\alpha$-Ham-Sandwich theorem generalizes to pizza cuttings.
For point sets in fixed dimensions, we give a linear time algorithm to find such a cut, whereas for arbitrary dimensions we place the problem in UEOPL.

Many of our results are more or less direct applications of known results.
In particular, all hardness results follow from a proof of the existence of a Consensus Halving from the existence of a pizza cut.
We present this proof and the hardness results in Section \ref{sec:hardness}.
In Section \ref{sec:separated}, we consider the case of well-separates masses and point sets, as some of the ideas are needed for our containment proof.
The part where the most new ideas are used is Section \ref{sec:containment}, where we show the containment of \textsc{DiscretePizzaCutting} in PPA.

\section{Hardness results}\label{sec:hardness}

In this section, we give a proof of existence of Consensus halvings and necklace splittings using the planar pizza cutting theorem.
This proof gives a natural reduction to the corresponding algorithmic problems, and thus a variety of hardness results follow.

\begin{lemma}
The planar pizza cutting theorem implies the Hobby-Rice theorem.
\end{lemma}

\begin{proof}
Consider the moment curve $\gamma$ in $\mathbb{R}^2$, that is, the curve parametrized by $(t,t^2)$.
Note that any line $\ell$ intersects $\gamma$ in at most two points, call them $g_1$ and $g_2$, with $g_1$ being to the left of $g_2$ (in the case of a single or no intersection, we may consider $g_1=-\infty$ or $g_2=\infty$ or both).
Let $x_1$ and $x_2$ be the projections of $g_1$ and $g_2$ to the $x$-axis under the projection $\pi(x,y):=x$.
Consider now a half-plane $h$ bounded by $\ell$ and consider its intersection with $\gamma$.
If $h$ lies below $\ell$, then $h\cap\gamma$ projects to the interval $[x_1,x_2]$.
If $h$ lies above $\ell$, then $h\cap\gamma$ projects to $(-\infty,x_1]\cup[x_2,\infty)$.
Similarly, if $\ell$ is vertical, $h\cap\gamma$ projects either to $(-\infty,x_1]$ or $[x_2,\infty)$.
Thus, in all cases $h\cap\gamma$ projects to an interval or the complement of an interval, which, in a slight abuse of notation, we denote by $\pi(h)$.

Given a valuation function $v$, we now want to define a mass distribution $\mu$ in the plane.
For this, it is enough to just define $\mu$ on all half-planes.
This we can do using the above observations: for any half-plane $h$, we define $\mu(h):=v(\pi(h))$.

This way, we have defined $n$ mass distributions.
Now, we do the same thing, but shift the interval $[0,1]$, which is the support of the valuation functions, to the interval $[2,3]$.
More formally, we consider the projection $\varphi(x,y):=x-2$ and, given a valuation function $v$, define a mass distribution $\eta$ as $\eta(h):=v(\varphi(h))$.

We have now defined $2n$ mass distributions.
By the pizza cutting theorem, there exists and arrangement $A=(\ell_1,\ldots,\ell_n)$ of $n$ lines which simultaneously bisects these mass distributions.
Consider the intervals $I_1$ and $I_2$ on $\gamma$ defined by $t\in[0,1]$ and $t\in[2,3]$.
As there are at most $2n$ intersections of $A$ and $\gamma$, by the pigeonhole principle there are at most $n$ intersections in one of them, say $I_1$.
Let $i_1,\ldots,i_n$ be the projections of these intersections under the projection $\pi$.
We claim that $i_1,\ldots,i_n$ simultaneously bisect $v_1,\ldots,v_n$.

To show this, consider some valuation function $v_i$.
By construction, we now have that $v_i(\mathcal{I}^+)=\mu_i(R^+(A))=\mu_i(R^-(A))=v_i(\mathcal{I}^+)$, which proves the claim.
In the case where $I_2$ has at most $n$ intersections, we can do the same argument, replacing $\pi$ with $\varphi$ and $\mu$ with $\eta$.
\end{proof}

Note that all the steps in the proof can be computed in polynomial time.
Thus, as \textsc{ConsensusHalving} is FIXP-hard \cite{ExactConsensus}, we immediately get the following:

\begin{corollary}
\textsc{PizzaCutting} is FIXP-hard.
\end{corollary}

In the discrete setting, the above proof can be phrased even simpler: for each point $x$ in $[0,1]$, we just define two points $(x,x^2)$ and $(x+2,(x+2)^2)$.
As \textsc{NecklaceSplitting} is PPA-hard \cite{GoldbergNecklace}, we get that

\begin{corollary}\label{cor:ppahard}
\textsc{DiscretePizzaCutting} is PPA-hard.
\end{corollary}

Clearly, the construction in the proof above also works for more than $n$ valuation functions.
In \cite{ExactConsensus}, it was shown that deciding whether $n+1$ valuation functions can be bisected with $n$ cuts is $\exists\mathbb{R}$-hard.
It thus follows that it is $\exists\mathbb{R}$-hard to decide whether $2n+2$ masses can be bisected by $n$ lines.
However, in the case where $n$ is even, we can also only use $\pi$ and map all valuation functions to a single interval on $\gamma$, which analogously proves the Hobby-Rice theorem for even $n$.
From an asymptotic point of view, this restriction to even values does not matter, so using this reduction, we get the following slightly stronger statement.

\begin{corollary}
\textsc{PizzaCuttingDecision} is $\exists\mathbb{R}$-hard.
\end{corollary}

Finally, it was shown in \cite{NecklaceNP1, NecklaceNP2} that finding the minimal number of cuts required to split a necklace is NP-hard.
We again get the analogous result for discrete pizza cuttings.

\begin{corollary}
Finding the minimal number of lines that simultaneously bisect a family of $2n$ point sets is NP-hard.
\end{corollary}

\section{Well-separated point sets}\label{sec:separated}

In this section we consider well-separated mass distributions and point sets.
We generalize the $\alpha$-Ham-Sandwich theorem to pizza cuttings.

\begin{theorem}\label{thm:alpha_pizza}
Let $\mu_1,\ldots,\mu_{nd}$ be $nd$ well-separated mass distributions in $\mathbb{R}^d$.
Given a vector $\alpha=(\alpha_1,\ldots,\alpha_{nd})$, with each $\alpha_i\in [0,1]$, there exists an arrangement $A$ of $n$ oriented hyperplanes such that for each $i\in\{1,\ldots,nd\}$ we have $\mu_i(R^+(A))=\alpha_i\mu_i(\mathbb{R}^d)$.
\end{theorem}

\begin{proof}
By the $\alpha$-Ham-Sandwich theorem \cite{AlphaHSMasses}, for any $d$ well-separated mass distributions $\mu_1,\ldots,\mu_d$ and any vector $(\alpha_1,\ldots,\alpha_d)$, there exists a unique single hyperplane cutting each mass distribution in the required ratio.
By the definition of well-separatedness, this hyperplane does not intersect the support of any other mass distribution.
Partition the point set into $n$ parts of $d$ point sets each.
For each part, pick some oriented hyperplane which intersects the support of all masses in this part.
This defines an arrangement $B$ of oriented hyperplanes.
For each mass $\mu_i$, consider the intersection of its support with the positive side of the hyperplane intersecting it, as well as with $R^+(B)$.
If these two coincide, set $\alpha'_i:=\alpha_i$, otherwise set $\alpha'_i:=1-\alpha_i$.
Now, taking the $\alpha$-Ham-Sandwich cut for the vectors $\alpha'$ gives the required arrangement.
\end{proof}

We call such an arrangement an \emph{$\alpha$-Pizza cut}.
From the discrete $\alpha$-Ham-Sandwich theorem \cite{AlphaHSPoints}, we analogously get the discrete version of the above.

\begin{corollary}
Let $P_1,\ldots,P_{nd}$ be $nd$ well-separated point sets in $\mathbb{R}^d$.
Given a vector $\alpha=(k_1,\ldots,k_{nd})$, where each $k_i$ is an integer with $0\leq k_i\leq |P_i|$, there exists an arrangement $A$ of $n$ oriented hyperplanes such that for each $i\in\{1,\ldots,nd\}$ we have $|P_i\cap R^+(A)|=k_i$.
\end{corollary}

In \cite{AlphaHS}, it was shown that the problem of computing an $\alpha$-Ham-Sandwich cut for point sets is in UEOPL.
As our $\alpha$-Pizza cuts are just a union of $\alpha$-Ham-Sandwich cuts, their result generalizes to our setting.

\begin{corollary}
The problem of computing an $\alpha$-Pizza cut for point sets is in UEOPL.
\end{corollary}

\begin{remark}
The computation of the vector $\alpha'$ does not directly translate into the setting of UEOPL.
However, in \cite{AlphaHS}, they start with an arbitrary hyperplane, and then rotate it in a well-defined fashion to the solution.
This immediately translates to our setting: we just start with an arbitrary arrangement, and the choice of the direction of rotation works analogously to the choice in \cite{AlphaHS}.
\end{remark}

Further, Bereg \cite{BeregAlpha} has shown that an $\alpha$-Ham-Sandwich cut for $d$ point sets of $m$ points total in $\mathbb{R}^d$ can be computed in time $m2^{O(d)}$.
In particular, if $d$ is fixed, this algorithm runs in linear time.
Again, we get the same result for $\alpha$-Pizza cuts.

\begin{corollary}
Let $P_1,\ldots,P_{nd}$ be well-separated points sets in $\mathbb{R}^d$ with $\sum_{i=1}^{nd}|P_i|=m$.
Then an $\alpha$-Pizza cut for $P_1,\ldots,P_{nd}$ can be computed in time $m2^{O(d)}$.
\end{corollary}

\begin{proof}
Partition the point sets into parts $P_{(i-1)d+1},\ldots,P_{id}$, for $i\in\{1,\ldots,n\}$.
Compute the vector $\alpha'$ in time $O(nd)$.
As $m\geq nd$, the runtime of the second part dominates the total runtime.
For each $i\in\{1,\ldots,n\}$, use Bereg's algorithm to compute the $\alpha$-Ham-Sandwich cut for $P_{(i-1)d+1},\ldots,P_{id}$.
It follows that the solution is an $\alpha$-Pizza cut.
The runtime of the algorithm is
\[\sum_{i=1}^n (|P_{(i-1)d+1}|+,\ldots,+|P_{id}|)2^{O(d)}=\sum_{i=1}^{nd}|P_i|2^{O(d)}=m2^{O(d)}.\]
\end{proof}

\section{Containment results}\label{sec:containment}

In this section, we prove that \textsc{DiscretePizzaCutting} is in PPA.
We do this by giving a new proof of the discrete planar pizza cutting theorem, which allows for an algorithm in PPA.
Before we go into the details of the proof, we briefly sketch the main ideas.
The main structure of our proof is similar to and inspired by the original proof of Hubard and Karasev \cite{Pizza1}, but we replace their topological arguments with easier ones that only use the combinatorics of point sets, but hence do not work for the more general case of mass distributions.

The main idea is that we continuously transform well-separated point sets into the point sets which we want to bisect.
In the beginning of this process, we have several bisections, namely one for each partition of the labels of the point sets into pairs, and this number is odd.
During the process, we pull these bisecting arrangements along.
Every time the orientation of some triple of points changes, it can happen that one of the arrangements is not bisecting anymore.
The main step in the proof is, that in these cases, we can always slightly change this arrangement so that it is again bisecting, or that there is a second arrangement that also is not bisecting anymore.
In other words, some bisections might vanish during the process, but if they do, then they always vanish in pairs.
This step is also where our proof differs from the one by Hubard and Karasev.

Once we have this, the remainder of the proof is rather simple: as we started with an odd number of bisections, and they always vanish in pairs, the number of bisections is always odd, and thus in particular at least 1.
Further, we can build a graph where each vertex corresponds to a point set in our process with an arrangement, where the vertices are connected whenever one arrangement is the pulled along version of the other one, or when both point sets are the same and the arrangements are the two arrangements that vanish at this point of the process.
(In the end, some of these connections will be paths instead of single edges.)
Adding an additional vertex which we connect to all starting arrangements, we get a graph in which the only odd-degree vertices are this additional vertex and the final solutions.

\subsection{A proof of the discrete planar pizza cutting theorem}

We now proceed to give a detailed proof of the discrete planar pizza cutting theorem (Corollary \ref{cor:discreteplanar}).
Let $P_1,\ldots,P_{2n}$ be the point sets we wish to bisect.
Let further $m:=\sum_{i=1}^{2n}|P_i|$.

\begin{lemma}\label{lem:odd}
We may assume that each $P_i$ contains an odd number of points.
\end{lemma}

\begin{proof}
For each point set $P_i$ with an even number of points, add some arbitrary new point $q_i$ to $P_i$, such that the point sets are still in general position.
Take some bisecting arrangement $A$ of the resulting point sets.
As all of these point sets consist of an odd number of points in general position, each of them must contain a point $p_i$ which lies on a line of $A$.
As each line in $A$ can pass through at most two points by the general position assumption, there is exactly one such point in each point set.
Now, remove $q_i$ again.
If $p_i=q_i$, the arrangement still bisects $P_i$.
Otherwise, one side, without loss of generality $R^+$, contains one point too few.
Rotate the line through $p_i$ slightly so that $p_i$ lies in $R^+$.
Now the arrangement again bisects $P_i$.
\end{proof}

So, from now on we may assume that each $P_i$ contains an odd number of points.
Let $Q_1,\ldots,Q_{2n}$ be point sets of the same size, that is, $|Q_i|=|P_i|$, that are well-separated.
Match each point $q\in Q_i$ with a point $p\in P_i$ and consider for each such pair the map $\varphi(t):=tp+(1-t)q$.
These maps define a map $\Phi(t)$, which assigns to each $t$ the point set defined by the $\varphi(t)$'s.

During this process, some orientations of triples of points change.
We may assume, that no two triples change their orientation at the same time.
Further, as each point crosses each line spanned by two other points at most once, the total number of orientation changes is in $O(m^3)$.
Thus, the interval $[0,1]$ is partitioned into $O(m^3)$ subintervals, in each of which the orientations of all triples stay the same, that is, the \emph{order type} is invariant.
At the boundaries of these subintervals, some three points are collinear.

Taking a representative of each subinterval and the point sets at their boundaries, we thus get a sequence of point sets
\[Q_i=P^{(0)}_i,P^{(0.5)}_i,P^{(1)}_i,\ldots,P^{(k)}_i,P^{(k.5)}_i,P^{(k+1)}_i,\ldots,P^{(C)}_i=P_i,\]
where $C\in O(m^3)$ is the number of orientation changes, $P^{(k)}:=\bigcup_{i=1}^{2n}P^{(k)}_i$ is a point set in general position, and $P^{(k.5):=}\bigcup_{i=1}^{2n}P^{(k.5)}_i$ is a point set with exactly one collinear triple.
We will mainly work with the sets $P^{(k)}$, the sets $P^{(k.5)}$ are just used to make some arguments easier to understand.

As argued before, for each $P^{(k)}$, any bisecting arrangement $A$ contains exactly one point $p^{(k)}_i(A)$ of each $P^{(k)}_i$ on one of its lines.
In particular, $A$ is defined by a set of pairs of points $(p^{(k)}_i(A),p^{(k)}_j(A))$, where each pair defines one line of the arrangement.

\begin{lemma}\label{lem:beginning}
For $P^{(0)}$, the number of bisecting arrangements is odd.
\end{lemma}

\begin{proof}
It is known that for 2 separated point sets, each of odd size, there is a unique Ham-Sandwich cut (see e.g.\ \cite{AlphaHSPoints}).
We have seen in the proof of Theorem \ref{thm:alpha_pizza}, that for well-separated point sets, each bisecting arrangement corresponds to $n$ Ham-Sandwich cuts, each for a pair of point sets.
Thus, the number of bisecting arrangements is the same as the number of partitions of $2n$ elements into pairs.
This number is $(2n-1)!!=(2n-1)(2n-3)\cdots 3\cdot 1$, which is a product of odd numbers, and thus odd.
\end{proof}

Let us now follow some arrangement $A$ through the process.
More precisely, Let $A^{(0)}$ be a bisecting arrangement for $P^{(0)}$, defined by pairs of points $(p^{(0)}_i(A),p^{(0)}_j(A))$.
We now consider the sequence of arrangements $A^{(k)}$, where each $A^{(k)}$ is defined by the corresponding pairs of points $(p^{(k)}_i(A),p^{(k)}_j(A))$.
Clearly, if $A^{(k)}$ is bisecting and the orientation change from $P^{(k)}$ to $P^{(k+1)}$ is not a point moving over a line in the arrangement $A^{(k)}$, then $A^{(k+1)}$ is still bisecting.
For the other changes, we need the following lemma.
Here, we say that an arrangement $A^{(k.5)}$ is \emph{almost bisecting} for $P^{(k.5)}$, if it bisects each $P_i^{(k.5)}$ except for one, for which two points are on a line of the arrangement, and for the remaining points one side contains exactly one point more.

\begin{lemma}\label{lem:main}
Let $A^{(k)}$ and $A^{(k+1)}$ be such that $A^{(k)}$ is bisecting and $A^{(k+1)}$ is not.
Then there either exists a sequence of arrangements
\[A^{(k)},A^{(k.5)}=A_0^{(k.5)},A_1^{(k.5)},\ldots,A_L^{(k.5)}=B^{(k.5)},B^{(k+1)},\]
where each $A_l^{(k.5)}$ is almost bisecting, $B^{(k+1)}$ is bisecting and $B^{(k)}$ is not bisecting, or there exists a sequence of arrangements
\[A^{(k)},A^{(k.5)}=A_0^{(k.5)},A_1^{(k.5)},\ldots,A_L^{(k.5)}=B^{(k.5)},B^{(k)},\]
where each $A_l^{(k.5)}$ is almost bisecting, $B^{(k)}$ is bisecting and $B^{(k+1)}$ is not bisecting.
Further, in the second case, the sequence for $B^{(k)}$ is the reverse of the sequence for $A^{(k)}$.
\end{lemma}

In this lemma, the first case corresponds to a bisecting arrangement that can be changed to a new bisecting arrangement, whereas the second case corresponds to two bisecting arrangements disappearing at the same time.

\begin{proof}
As mentioned above, the situation of the lemma can only occur if the orientation change from $P^{(k)}$ to $P^{(k+1)}$ corresponds to a point $q^{(k)}$ moving over a line $\ell$ of the arrangement $P^{(k)}$.
Without loss of generality, let $\ell$ be defined by the points $p^{(k)}_1\in P_1^{(k)}$ and $p^{(k)}_2\in P_2^{(k)}$.
There are two cases we consider: either $q^{(k)}$ is in the same point set as one of the points on $\ell$, without loss of generality $q^{(k)}\in P_1^{(k)}$, or it is in a different point set, without loss of generality $q^{(k)}\in P_3^{(k)}$.
We start with the second case.
For an illustration of that case, see Figure \ref{fig:flip}.

\begin{figure}
\centering
\includegraphics[scale=0.85]{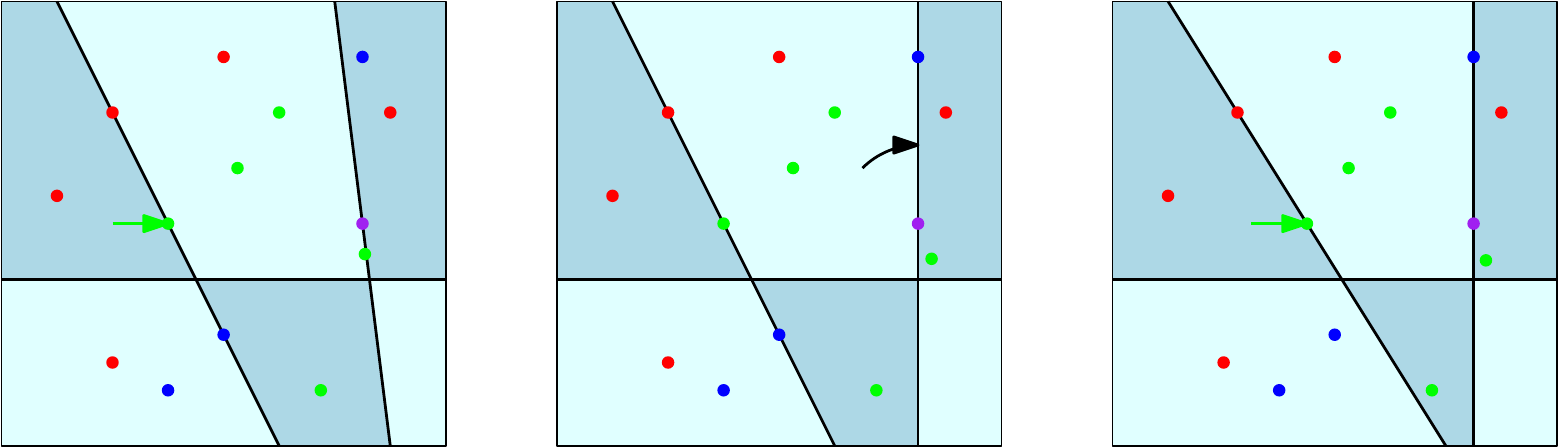}
\caption{Going from $A^{(k.5)}$ (left) via $B^{(k.5)}$ (middle) to $B^{(k+1)}$ (right). In this example, $B^{(k+1)}$ does not bisect the blue point set anymore.}
\label{fig:flip}
\end{figure}

\textbf{Case 1: $q^{(k)}\in P_3^{(k)}$.}
Orient the arrangement in such a way that $q^{(k)}$ is in $R^+(A^{(k)})$.
In $A^{(k.5)}$ the line $\ell$ contains three points, namely $p^{(k)}_1$, $p^{(k)}_2$ and $q^{(k)}$.
Note that $A^{(k.5)}$ is almost bisecting, and $R^+(A^{(k.5)})$ is the smaller side for $P_3^{(k.5)}$.
There exists another line $\ell'$ in $A^{(k.5)}$ which contains a second point $p_3^{(k.5)}$ of $P_3^{(k.5)}$.
Rotate $\ell'$ around the other point on it such that $p_3^{(k.5)}$ is in $R^+(A^{(k.5)})$.
Note that this direction of rotation is unique.
Continue the rotation until $\ell'$ hits another point $q_j^{(k.5)}\in P_j^{(k.5)}$ which is not on some line of $A^{(k.5)}$.
The resulting arrangement is now $A_1^{(k.5)}$.

Note that $A_1^{(k.5)}$ is again almost bisecting, and that the point set which is not exactly bisected is $P_j^{(k.5)}$.
While $j\not\in\{1,2,3\}$, we can now find another point in $P_j^{(k.5)}$ lying on a line of the arrangement $A_l^{(k.5)}$, rotate this line into the correct direction until a new point is hit, to get a new arrangement $A_{l+1}^{(k.5)}$.
As all of these arrangements are different, and there are only finitely many possible arrangements, at some point we will have $j\in\{1,2,3\}$.
This gives the arrangement $B^{(k.5)}$.
There are now several options to consider: there are 3 different ways how $p^{(k)}_1$, $p^{(k)}_2$ and $q^{(k)}$ can lie on $\ell$, $j$ can be $1$, $2$ or $3$, and the smaller side of $P_j^{(k.5)}$ can be $R^+(B^{(k.5)})$ or $R^-(B^{(k.5)})$.
In all of the cases, the only change from the arrangements $B^{(k)}$ to $B^{(k+1)}$ are the defining points of $\ell$.
It follows that exactly one of $B^{(k)}$ and $B^{(k+1)}$ can be bisecting, depending on whether the point on $\ell$ is in $R^+(B^{(k+1)})$ or not.

\textbf{Case 2: $q^{(k)}\in P_1^{(k)}$.}
In this case, we just have $A^{(k.5)}=B^{(k.5)}$.

Finally, the last claim follows from the fact that the directions of rotations from $A_l^{(k.5)}$ to $A_{l+1}^{(k.5)}$ are unique.
\end{proof}

\begin{remark}
During the process, it can also happen that two bisecting arrangement appear at the same time.
In this case, we immediately get the analogous lemma, just reversing $k$ and $(k+1)$.
\end{remark}

With all these lemmas at hand, we can now finally give a

\begin{proof}[Proof of the discrete planar pizza cutting theorem]
Let $P=(P_1,\ldots,P_{2n})$ be in general position.
By Lemma \ref{lem:odd}, we may assume that each $P_i$ contains an odd number of points.
Consider the sequence $P^{(0)},P^{(1)},\ldots,\ldots,P^{(C)}=P$ of point sets.
By Lemma \ref{lem:beginning}, $P^{(0)}$ has an odd number of bisecting arrangements.
When going from $P^{(k)}$ to $P^{(k+1)}$, either all bisecting arrangements stay bisecting or, by Lemma \ref{lem:main}, one of them can be changed into a new bisecting arrangement or exactly two bisecting arrangements appear or disappear.
It follows that for each $k$, $P^{(k)}$ has an odd number of bisecting arrangements.
In particular, also $P$ has an odd number of bisecting arrangements, and thus at least 1.
\end{proof}

\subsection{Containment in PPA}

In order to show that \textsc{DiscretePizzaCutting} is in PPA, we need to define a graph where the neighborhoods of each vertex can be efficiently computed and all odd-degree vertices, except the starting vertex, correspond to bisecting arrangements.
In the following, we describe such a graph.
For an illustration of the graph, see Figure \ref{fig:ppagraph}.

\begin{figure}
\centering
\includegraphics[scale=0.8]{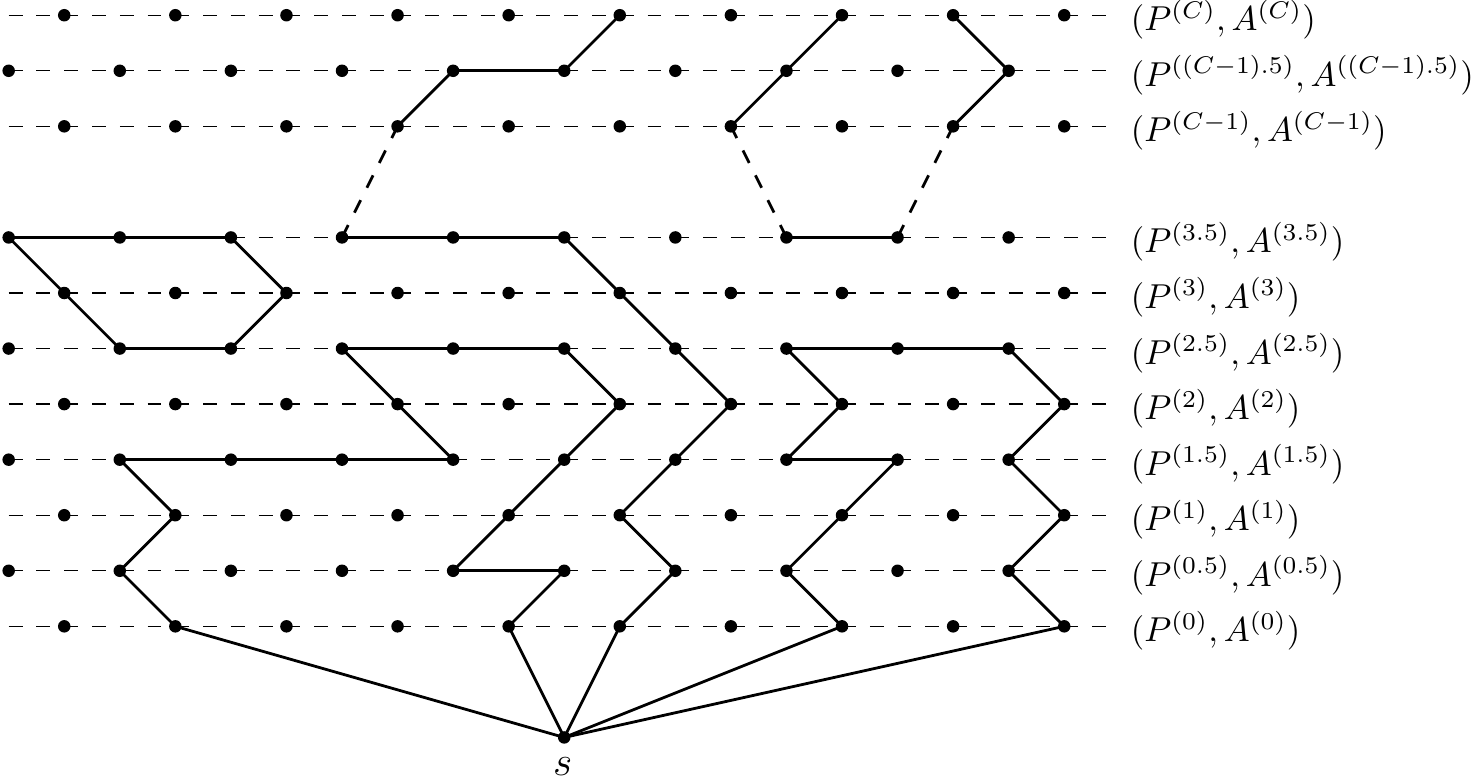}
\caption{A schematic drawing of the graph whose leafs correspond to bisecting arrangements.}
\label{fig:ppagraph}
\end{figure}

\textbf{The vertex set.} Our vertex set consists of a starting vertex $s$, as well as vertices of the form $(P^{(k)},A^{(k)})$ and $(P^{(k.5)},A^{(k.5)})$.
For vertices $(P^{(k)},A^{(k)})$, $A^{(k)}$ is an arrangement which is not necessarily bisecting, but whose lines contain exactly one point of each $P_i^{(k)}$.
Similarly, for vertices $(P^{(k.5)},A^{(k.5)})$, the arrangement $A^{(k.5)}$ is not necessarily almost bisecting, but its lines contain at least one point of each $P_i^{(k.5)}$ and exactly 2 for one of them.
In particular, one of the lines of $A^{(k.5)}$ is the line $\ell$ through the unique three collinear points.

\textbf{The edge set.} The starting vertex $s$ in connected to all vertices $(P^{(0)},A^{(0)})$, for which $A^{(0)}$ is bisecting.
A vertex of the form $(P^{(k)},A^{(k)})$, $k<L$ is connected to the vertices $(P^{((k-1).5)},A^{((k-1).5)})$ and $(P^{(k.5)},A^{(k.5)})$ if and only if $A^{(k)}$ is bisecting.
Otherwise, it is not connected to any other vertex.
Similarly, the vertices $(P^{(C)},A^{(C)})$ are connected to $(P^{((C-1).5)},A^{((C-1).5)})$ if and only if $A^{(C)}$ is bisecting, and not connected to anything otherwise.
Also a vertex of the form $(P^{(k.5)},A^{(k.5)})$ where $A^{(k.5)}$ is not almost bisecting is not connected to any other vertex.
For vertices $(P^{(k.5)},A^{(k.5)})$ with $A^{(k.5)}$ almost bisecting we distinguish two cases.
If $\ell$ does not contain a point of the point set which is contained twice in the lines of $A^{(k.5)}$, we connect $(P^{(k.5)},A^{(k.5)})$ to the two vertices which correspond to the rotations defined in the proof of Lemma \ref{lem:main}.
Otherwise, we connect $(P^{(k.5)},A^{(k.5)})$ to the single vertex that we get by such a rotation, as well as to either $(P^{(k)},A^{(k)})$ or $(P^{(k+1)},A^{(k+1)})$, depending on which one has the bisecting arrangement (note that we get from Lemma \ref{lem:main} that exactly one of the two arrangements is bisecting).

Note that all vertices have degree 0 or 2, except $s$ and the vertices $(P^{(C)},A^{(C)})$ with $A^{(C)}$ bisecting, that is, the vertices corresponding to solutions.
All the vertices corresponding to solutions have degree 1.
The starting vertex has degree $(2n-1)!!$, which is exponential, so we cannot compute its neighborhood in polynomial time.
This is however not an issue, as already noted in the paper where PPA is defined \cite{Papadimitriou}, provided that we can decide for any two vertices in polynomial time whether the are connected, and there is an efficiently computable \emph{pairing function}, that is, a function $\phi(v,v')$ which takes as input a vertex $v$ and an outgoing edge $\{v,v'\}$ and computes another vertex $v''$ such that $\{v,v''\}$ is also an edge.
If $v$ has even degree, then $\phi(v,v')\neq v'$ for all $v'$.
If $v$ has odd degree, then there is exactly one $v'$ for which $\phi(v,v')= v'$.
The exact statement from \cite{Papadimitriou} is as follows:

\begin{lemma}[\cite{Papadimitriou}]\label{lem:ppa}
Any problem defined in terms of an edge recognition algorithm and a pairing function is in PPA.
\end{lemma}

It remains to show that we have both these ingredients.

\begin{lemma}
For any two vertices, we can decide in polynomial time whether they are connected.
\end{lemma}

\begin{proof}
It follows from the construction that for any vertex except the starting vertex $s$, the neighborhood can be computed in polynomial time.
In particular, it can also be checked whether two such vertices are connected.
As for the starting vertex $s$, some other vertex can only be connected to it if its point set is $P^{(0)}$ and its arrangement is bisecting.
The first is encoded in the label of the vertex and the second can easily be checked in polynomial time.
\end{proof}

\begin{lemma}
There is a pairing function $\phi$ which can be computed in polynomial time.
\end{lemma}

\begin{proof}
As all vertices except $s$ have degree $0$, $1$ or $2$, and the neighborhoods can be computed in polynomial time, the pairing function follows trivially for these vertices.
As for $s$, we note that its neighbors correspond to partitions of the $2n$ point sets into pairs.
In other words, its neighbors can be interpreted as perfect matchings in the complete graph $K_{2n}$ with vertex set $\{w_1,\ldots,w_{2n}\}$.
It is thus enough to describe a pairing function for these perfect matchings.
We do this in an algorithmic fashion.

Let $M$ be some perfect matching.
Consider the vertices $w_1$ and $w_2$.
Assume first that they are not connected to each other, that is, we have two distinct edges $\{w_1,w_a\}$ and $\{w_2,w_b\}$ in $M$.
In that case, define the pairing $M':=M\setminus\{\{w_1,w_a\},\{w_2,w_b\}\}\cup\{\{w_1,w_b\},\{w_2,w_a\}\}$, that is, flip the edges incident to $w_1$ and $w_2$.
Note that the pairing of $M'$ is again $M$, that is, it indeed is a pairing.

If $w_1$ and $w_2$ are connected to each other, repeat the same process with $w_3$ and $w_4$, and so on.
This process defines a pairing for each matching except $\{\{w_1,w_2\},\{w_3,w_4\},\ldots,\{w_{2n-1},w_{2n}\}\}$.
Further, the algorithm clearly runs in polynomial time.
We thus get a valid pairing function.
\end{proof}

We thus get from Lemma \ref{lem:ppa} that \textsc{DiscretePizzaCutting} is in PPA.
Together with Corollary \ref{cor:ppahard}, we can thus conclude our main result.

\begin{corollary}
\textsc{DiscretePizzaCutting} is PPA-complete.
\end{corollary}

\section{Conclusion}

We have shown several complexity results related to the pizza cutting problem.
Our main result is that \textsc{DiscretePizzaCutting} is PPA-complete.
While we have only considered the planar case, the problem as well as our arguments extend to higher dimensions.
More precisely, the proof of the Hobby-Rice theorem using the pizza cutting theorem can be adapted to work for any dimension of the pizza cutting theorem.
Further, the proof of the discrete pizza cutting theorem can be adapted to any dimension where the number of initial solutions is odd.
This is the case if the dimension is a power of 2 \cite{Pizza1}.
Thus, in all these dimensions the analogous versions of \textsc{DiscretePizzaCutting} are also PPA-complete, assuming the dimension is fixed.
On the other hand, it is an open problem whether the pizza cutting theorem also holds in other dimensions.

We have also shown that \textsc{PizzaCutting} is FIXP-hard.
It is an interesting problem whether the problem is in FIXP or even harder.
One way to show that the problem is in FIXP is to find a proof for the planar pizza cutting theorem using Brouwer's fixpoint theorem.
Such a proof would have the potential to generalize to any dimension, resolving the above question.



\bibliographystyle{plainurl}
\bibliography{refs}

\end{document}